%% file: main.tex
\documentclass{article}

\usepackage[T1]{fontenc} %
\usepackage[utf8]{inputenc} %
\usepackage{anyfontsize}
\usepackage{ismir,amsmath,amsthm,cite,url}
\usepackage{graphicx}
\usepackage{color}
\usepackage{booktabs}       %
\usepackage{nicefrac}       %
\usepackage{microtype}      %
\usepackage{xcolor}         %
\usepackage{booktabs}
\usepackage{stfloats}
\usepackage{float}
\usepackage[export]{adjustbox}
\usepackage{changepage}
\usepackage{footnote}
\makesavenoteenv{tabular}
\usepackage{tabularx}
\usepackage{hhline}
\usepackage{subcaption}
\usepackage{enumitem}
\usepackage{mathtools}
\usepackage{lipsum,afterpage,refcount}
\usepackage{titlesec}
\newcommand{\setfootnotemark}{%
  \refstepcounter{footnote}%
  \footnotemark[\value{footnote}]}

\newlength{\myeqskip}  
\AtBeginDocument{%
    \setlength\abovedisplayskip{\myeqskip}%
    \setlength\belowdisplayskip{\myeqskip}%
    \setlength\abovedisplayshortskip{\myeqskip-\baselineskip}%
    \setlength\belowdisplayshortskip{\myeqskip}

    \setcounter{totalnumber}{60}
    \setcounter{topnumber}{60}
    \setcounter{bottomnumber}{60}
    \setcounter{dbltopnumber}{60} 
    
    \setlength\textfloatsep{3.0pt}
    
\titlespacing*{\section}{0pt}{1.0\baselineskip}{0.5\baselineskip}
\titlespacing*{\subsection}{0pt}{0.5\baselineskip}{0.5\baselineskip}
    }

\usepackage{caption}
\usepackage{subcaption}

\providecommand{\paragraph}{}
\renewcommand{\paragraph}[1]{\noindent{\normalsize\bf #1}{\nobreak}}

\usepackage{lineno}
\input{math_commands}

\usepackage[capitalise]{cleveref}

\newtheorem{theorem}{Theorem}[section]

\newtheorem{lemma}[theorem]{Lemma}
\theoremstyle{definition}
\newtheorem{definition}{Definition}[section]

\title{Scoring {Time} Intervals using Non-Hierarchical Transformer for Automatic Piano Transcription}

\oneauthor
{Yujia Yan, Zhiyao Duan}
{
University of Rochester,
Rochester, New York, USA,\\
yujia.yan@rochester.edu, zhiyao.duan@rochester.edu}

\def\authorname{Y. Yan, and Z. Duan}

\begin{document}

\maketitle
\begin{abstract}
The neural semi-Markov Conditional Random Field (semi-CRF) framework has demonstrated promise for event-based piano transcription. 
In this framework, all events (notes or pedals) are represented as closed {time} intervals tied to specific event types.
The neural semi-CRF approach requires an interval scoring matrix that assigns a score for every candidate interval.
However, designing an efficient and expressive architecture for scoring intervals is not trivial.
This paper introduces a simple method for scoring intervals using scaled inner product operations that resemble how attention scoring is done in transformers.
We show theoretically that, due to the special structure from encoding the non-overlapping intervals, under a mild condition, the inner product operations are expressive enough to represent an ideal scoring matrix that can yield the correct transcription result. 
We then demonstrate that an encoder-only non-hierarchical transformer backbone, operating only on a low-time-resolution feature map, is capable of transcribing piano notes and pedals with high accuracy and time precision. 
The experiment shows that our approach achieves the new state-of-the-art performance across all subtasks in terms of the F1 measure on the Maestro dataset.
\end{abstract}

\noindent \textbf{See appendix for post-camera-ready updates.}

\section{Introduction}

Automatic Music Transcription (AMT) transforms the audio signal of music performances into symbolic representations \cite{Benetos2019AutomaticMT}. 
In this work, we focus on transcribing piano performance audio into its piano roll representation.\footnote{Code: \url{https://github.com/Yujia-Yan/Transkun}}
The piano roll representation, as formulated in \cite{yan2021skipping}, can be abstracted as consisting of sets of non-overlapping {time} intervals of the form $[\textrm{onset}, \textrm{offset}]$, with each set corresponding to one particular event type, e.g., a specific note or pedal.

Recent strategies to handle the problem of outputting this structured representation fall into three main categories:
1) Keypoint detection and assembly:
This approach involves identifying the onsets, offsets, and frame-wise activations of notes and then assembling these elements together with a handcrafted post-processing step.
Examples include~\cite{Hawthorne2018Onset, Kong2020HighResolutionPT, Toyama2023AutomaticPT}; 
2) Structured prediction with a probabilistic model:
Models in this category use a probabilistic model to ensure the structure of the output to be sets of non-overlapping intervals, 
e.g., \cite{yan2021skipping, kelz2019deepadsr, kwon2020polyphonic};
3) Sequence-to-sequence (Seq2Seq) methods\footnote{Strictly speaking, the Seq2Seq approach can also be categorized as a probabilistic model for structrued prediction. We isolate it here for simplifying the discussion.}:
These methods, such as \cite{Hawthorne2021SequencetoSequencePT}, treat music transcription as a machine translation problem, which translates audio to tokens that encode the target symbolic representation.

Our study focuses on the neural semi-Markov Conditional Random Field (semi-CRF) framework~\cite{yan2021skipping} from the second category, which directly models each music event (note or pedal) as a closed time interval associated with a specific event type.
The approach employs a neural network to score interval candidates and uses dynamic programming to decode non-overlapping intervals. 
This framework eliminates the need for separate keypoint detection and assembly steps in the first category but outputs the events (intervals) in a single stage.
Compared to other methods in the second category, e.g. \cite{kelz2019deepadsr, kwon2020polyphonic}, it does not need hand-crafted state definitions and state transitions.
Additionally, it benefits from optimal decoding in a non-autoregressive fashion as opposed to the slow autoregressive and suboptimal decoding in Seq2Seq methods (the third category).

This paper builds upon, simplifies, and improves the neural semi-CRF framework \cite{yan2021skipping} for piano transcription.
Our major contributions are as follows.
First, we replace the original scoring module that assigns a score for every possible interval with a simpler and more efficient pairwise inner product operation. 
Specifically, we prove that due to the special structure of encoding non-overlapping intervals, under a mild condition, the inner product operation is expressive enough to represent an ideal scoring matrix that can yield the correct transcription decoding. 
Second, inspired by the resemblance between the proposed inner product operation and the attention mechanism in the transformer \cite{Vaswani2017AttentionIA}, we use the transformer architecture to produce the interval representation for inner product scoring.
We demonstrate that an encoder-only non-hierarchical transformer backbone, operating only on a low-time-resolution feature map, is capable of transcribing notes with high accuracy and time precision. 
Third, we compare our method against state-of-the-art piano transcription systems on the Maestro v3 dataset, showing that our method establishes the new state of the art across all subtasks in terms of the F1 score.

\section{Related Work}

\subsection{Neural Semi-CRF for Piano Transcription}
Previous work of \cite{yan2021skipping} introduced a neural semi-Markov Conditional Random Field (semi-CRF) framework for event-based piano transcription, where each event (note or pedal) is represented as a closed interval associated with a specific event type. 
The approach employs a neural network to score interval candidates and uses dynamic programming to decode non-overlapping intervals. 
After interval decoding, interval-based features are used to estimate event attributes, such as \textit{MIDI velocity} and \textit{refined onset/offset positions}\footnote{For dequantizing onset/offset positions from quantized positions.}.

The neural semi-CRF can be viewed as a general output layer, similar to a softmax layer, but tailored for handling non-overlapping intervals. 
For a sequence of $T$ frames, let $\gY$ denote a set of non-overlapping closed intervals. 
The semi-CRF layer for $\gY$ takes two inputs for each event type:
\begin{enumerate}[nosep]
\item 
$score(i, j)$: A $T \times T$ triangular matrix that scores every candidate interval $[i,j]$ for inclusion in $\gY$. The diagonal values $score(i,i)$ represent single-frame events.
\item
$score_{\epsilon}(i-1, i)$: A $(T-1)$-dimensional vector that assigns a score to every interval $[i-1, i]$ not covered by any interval in $\gY$, serving as an inactivity score.

\end{enumerate}

Both $score(i, j)$ and $score_{\epsilon}(i-1,i)$ are computed using a neural network from the audio input $\gX$. The total score for $\gY$, given $\gX$, is:
\begin{equation}
\Phi(\gY| \gX) = \sum_{ [i,j] \in \gY } score(i,j) + \mkern-0mu\sum_{\substack{[i-1, i] \\\textrm{ not covered} \\\textrm{in } \gY }} score_{\epsilon}(i-1, i).
\end{equation}
For inference, maximum a posteriori (MAP) is used to infer the optimal set of non-overlapping intervals $\gY^*$:
\begin{equation}
\gY^* = \argmax_{\gY} \Phi(\gY| \gX).
\label{eqn:semiCRFInfer}
\end{equation}
For training, the maximum likelihood approach is used, with the conditional log-likelihood defined as:
\begin{equation}
\log p(\gY| \gX) = \Phi(\gY|\gX) - \log \sum_{\gY'} \exp \Phi(\gY'| \gX).
\label{eqn:semiCRFLikelihood}
\end{equation}
Here, $\argmax$ in \cref{eqn:semiCRFInfer}, and the summation in the second term in \cref{eqn:semiCRFLikelihood} are over all possible sets of non-overlapping intervals.
We refer the readers to \cite{yan2021skipping} for algorithmic details.

To make predictions for all event types (88 keys + pedals), multiple instances of semi-CRF are used in parallel, each corresponding to a specific event type.

\subsection{Vision Transformer and YOLOS}

The Vision Transformer (ViT) \cite{dosovitskiy2021an} introduced a significant shift in computer vision, offering an alternative to traditional CNN models. 
ViT processes images as sequences of fixed-size patches using transformer layers \cite{Vaswani2017AttentionIA}, proving successful across various tasks. 
{
For end-to-end object detection, YOLOS \cite{Fang2021YouOL} demonstrated a minimal, non-hierarchical encoder-only design that appends [DET] tokens (representing object slots) directly to image patch tokens as input to the transformer encoder. 
Our architecture adopts a similar encoder-only design for event-based music transcription.
}

\section{Revisiting Interval Scoring for Semi-CRFs}
The neural semi-CRF framework crucially relies on modeling the interval scoring matrix, $score(i,j)$, which assigns a score to each candidate interval.
The size of the matrix, which grows quadratically with the sequence length, poses a challenge to designing an efficient and expressive model architecture. 
For this discussion, $score_{\epsilon}$ will be excluded due to its minimal impact on model performance from our observation and negligible modeling challenges.

\subsection{Interval Scoring in \cite{yan2021skipping}}
\label{sec:originalScoring}
In \cite{yan2021skipping}, a backbone model first transforms the input sequence $\gX =  [\vx_0, \ldots, \vx_{T-1}]$ into a sequence of feature vectors $ [\vh_0, \ldots, \vh_{T-1}]$.
Each interval $[i, j]$ is scored by applying an MLP to features computed from the interval, with the output dimension being the number of event types.
For simplicity, assuming only one event type to predict, the score is computed as
\begin{equation}
\label{eqn:originalScoring}
   score(i,j) = \textit{MLP}([\vh_i, \vh_j, \vh_i\odot \vh_j, \vm_1, \vm_2, \vm_3]),
\end{equation}
where $\vh_i$ and $\vh_j$ are feature vectors corresponding to the interval's onset and offset, $\odot$ denotes element-wise multiplication, and $\vm_1, \vm_2, \vm_3$ are the first, second, and third statistical moments over the interval $[i, j]$.

After producing the initial interval scoring matrices for all event types, a shallow CNN is applied, treating the interval endpoints as spatial coordinates and event types as channels. 
This refinement step slightly improves the result. 

Directly computing \cref{eqn:originalScoring} and the subsequent refinement step are memory intensive.
The official implementation processes the scoring matrix in segments and applies gradient checkpointing during training, reducing peak memory usage at the cost of increased computational time. 
Consequently, the MLP and CNN layers' depth and width are constrained, potentially limiting the model's capacity and increasing susceptibility to local pattern overfitting.

\subsection{Interval Scoring with Inner Product} \label{sec:innerProductScoring}
We propose to use the following method for interval scoring:
\begin{equation}
score(i,j) = \frac{|j-i|}{\sqrt{D}}\langle\vq_i, \vk_j \rangle + b_i \delta(i, j),
\label{eqn:innerProductScoring}
\end{equation}
where $\delta(i,j)$ is the Kronecker delta, which is $1$ if $i = j$ and $0$ otherwise. 
$\vq_i \in \R^{D}$, $\vk_i \in \R^{D}$ and $b_i \in \R$ are computed from the embedding vector $\vh_i$ using a linear layer $f$:
\begin{equation}
   [\vq_i, \vk_i, b_i]  = f(\vh_i). 
   \label{eqn:kqbMapping}
\end{equation}
The interval scoring matrix computed from \cref{eqn:innerProductScoring} takes a low-rank plus diagonal structure.
This method, termed \newterm{Scaled Inner Product Interval Scoring}, computes the score of an event as the scaled inner product between vectors $\vq_{i}$ and $\vk_j$ representing the start and the end of the interval.

Despite its simplicity and resemblance to the attention mechanism in transformers, one question arises about the expressiveness of the inner product for capturing the transcription result.
We answer this question by constructing a family of interval scoring matrices that can yield the correct decoded result, and then show that this family of matrices can be represented in the form of pairwise inner product under certain conditions.

Without loss of generality, we ignore the intervals of form $[i,i]$, which correspond to the diagonal values in the interval scoring matrix; they can be added back as diagonals as in \cref{eqn:innerProductScoring}.
Additionally, since only the upper triangular part of the interval scoring matrix is used, we use the notation for a full matrix to simplify the derivation.
We begin by defining a set of nonoverlapping closed intervals.

\begin{definition} \label{def:nonoverlappingIntervals}
Let $\gY$ be a set of closed intervals defined on 
$\sN \cap [0, T-1]$, i.e., $T$ steps.
It is a set of \textit{non-overlapping} intervals if for any two intervals $[i_0, j_0] \in \gY$ and $[i_1, j_1] \in \gY$, $i_0\ge j_1$ or $i_1 \ge j_0$,
 and, additionally, $\forall [i, j] \in \gY, i< j$.
\end{definition}

\begin{definition}
   An ideal interval scoring matrix for $\gY$ over $T$ steps, i.e., $\mS_{\gY}\in \R^{T\times T}$, is a matrix such that 
   \begin{equation*}
   \begin{aligned}
    &\mS_{\gY}(i,j) > 0,  \quad&
    \forall [i,j]\in \gY,  \\
    &\mS_{\gY}(i,j) = -\epsilon,& \textit{ otherwise} 
    \end{aligned}
   \end{equation*}
   where $\epsilon>0$.
\end{definition}

With an ideal scoring matrix $\mS_{\gY}$, it is clear that the MAP decoding will yield $\gY$, since the exclusion of $\forall [i,j]\in \gY$ or the inclusion of $\forall [i,j] \notin \gY$ will decrease the total score.

\begin{lemma}
\label{lemma:idealRank}
    The rank of an ideal interval scoring matrix $S_{\gY}$ for a set of non-overlapping intervals, $\gY$, is $M+1$, where $M = |\gY|$, which is the number of intervals.
\end{lemma}
\begin{proof}

By definition, the first column is $-\epsilon\vone$, that is, $\forall i, \mS_{\gY}(i , 0) = -\epsilon$.
Subtracting the first column from all columns gives $\mS'_{\gY}$ such that
\begin{equation*}
\begin{aligned}
&\mS'_{\gY}(i,j) > \epsilon,  \quad&
\forall [i,j]\in \gY,  \\
&\mS'_{\gY}(i,j) = 0, & \textit{otherwise} 
\end{aligned}
\end{equation*}
Given that no two non-zero entries in $\mS'_{\gY}$ share a row or column (as per the definition of set of non-overlapping intervals), and there are $M$ non-zero entries, the rank of $\mS'_{\gY}$ is $M$.
Since there are at most $T-1$ non-overlapping intervals across $T$ frames, we have $M \le T-1$, and the number of nonzero entries in $S'_{\gY}$ is smaller than or equal to $T-1$.
As a result, $-\epsilon \vone$ ($T$ non-zeros) cannot be represented by a linear combination of other nonzero columns in $\mS'_{\gY}$, therefore ${rank(\mS_{\gY}) = rank(\mS'_{\gY})+1 = M+1}$.
\end{proof}

\begin{theorem}
\label{thm:intervalInnerProduct}
   Let $\gY$ be a set of non-overlapping closed intervals over $T$ steps, with cardinality $M$. 
   An ideal interval scoring matrix $\mS_{\gY}$ can be represented as pairwise inner products between two 1d sequences $(\vk_i)_i$ and $(\vq_i)_i$ of vectors:
   
   \begin{equation}
        \mS_{\gY}(i, j) =  
        \langle\vq_i, \vk_j \rangle,
   \end{equation}
 provided that  
  $rank(\mQ_{\gY}) > M $ and $rank(\mK_{\gY})>M$
  where
  $\mQ_{\gY} = [\vq_0, \ldots, \vq_{T-1}]$, and 
  $\mK_{\gY} = [\vk_0, \ldots, \vk_{T-1}]$. 
\end{theorem}

\begin{proof}
By \cref{lemma:idealRank}, the rank of $\mS_{\gY}$  is $M+1$. 
Then it directly follows the rank factorization of a matrix.
\end{proof}

\Cref{thm:intervalInnerProduct} establishes a minimum rank requirement for $\mQ_{\gY}$ and $\mathbf{K}_{\gY}$ to represent an ideal scoring matrix.
This leads to two key observations:
\begin{enumerate}[nosep]
\item
 The vector dimensions $D$ of $\vk_i$ and $\vq_i$ must exceed the total number of intervals, $|\gY|$.

\item 
Consider a linear upsampling operator $u_c$, which is a special case of a 1-d transposed convolutional layer. 
It works by dividing each step of a vector sequence into $c$ equal parts when the sequence is upsampled $c$ times.
Suppose we want to represent $\mQ_\gY$ and $\mK_{\gY}$ using low-resolution 1-d vector sequences: 
$\mQ'_{\gY} = [\vq'_0, \ldots, \vq'_{T'-1}]$ and 
$\mK'_{\gY} = [\vk'_0, \ldots, \vk'_{T'-1}]$
where $T'< T$, 
and this representation is achieved by applying $u_c$ to $\mQ'_{\gY}$ and $\mK'_{\gY}$, resulting in
$\mQ_\gY = u_c(\mQ'_\gY) $, and $\mK_\gY = u_c(\mK'_\gY) $, where $c=T/T'$ represents the upsampling factor.
For this representation to be valid, the vector dimension $D'$ for the low-resolution sequence, i.e., $\vq'_i$ and $\vk'_i$ should exceed $c |\gY|$.

\end{enumerate}
These observations highlight that the dimensionality requirement depends solely on the count of intervals in $\mathcal{Y}$ and the downsampling (upsampling) factor $c = T/T'$ along the time axis.
This analysis reveals sufficient conditions to guarantee the expressiveness of the inner product interval scoring method.
From \cref{thm:intervalInnerProduct}, by applying a scaling factor\footnote{Note that applying a length-dependent scaling on the ideal scoring matrix does not change the decoded result.} and reintegrating diagonal terms, we can recover \cref{eqn:innerProductScoring}.

\subsection{Comparison with Attention Mechanism}
\label{sec:connectionAttention}
Comparing the neural semi-CRF with the inner product scoring to the attention mechanism reveals interesting parallels. 
Both of them have quadratic time complexity in the length of the input.
The original score module, as in \cite{yan2021skipping}, resembles an additive attention mechanism, as introduced by \cite{bahdanau2015neural}.
However, attention mechanisms based on inner products \cite{luong2015effective} have become preferred for their simplicity and computational efficiency. 
Similarly, the proposed inner product scoring for neural semi-CRFs efficiently scores intervals.
However, in contrast to attention mechanisms that score sequence positions and normalize posteriors for each position, neural semi-CRFs score intervals and normalize posteriors globally over sets of non-overlapping intervals.

The Transformer architecture can be viewed as inherently refining a sequential representation for inner product scoring.
Inspired by these similarities, we utilize the transformer architecture to produce the 1-d sequence representations $(\vh_{i}^{\textrm{eventType}})$ for each event type, termed \newterm{event tracks}, which will be used for inner product interval scoring.

\section{Proposed System} \label{sec:proposedSystem}
\begin{figure}[htb]
    \centering
    \includegraphics[width=0.9\linewidth, trim={0 0.30cm 0.9cm 0.0cm},clip]{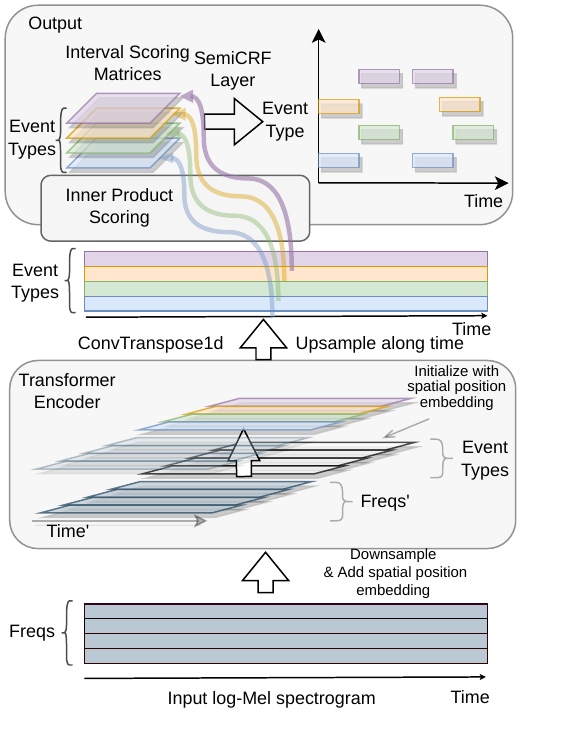}
    \caption{Overview of the proposed system. {Inner product scoring follows \cref{eqn:innerProductScoring}.}}
    \label{fig:overview}
\end{figure}

Figure 1 summarizes the proposed system.
The input is an oversampled log-mel spectrogram, as in \cite{yan2021skipping}.
The spectrogram is downsampled using {2-d} strided convolutional layers, followed by the addition of spatial position embeddings (\cref{sec:encoder}). 
Event tracks for all event types (notes and pedals) are initialized with their own spatial position embeddings and concatenated with the downsampled spectrogram representations.
The concatenated features are processed by a transformer encoder. 
Subsequently, only the event track embeddings are upsampled using one 1-d transposed convolutional layer. 
The upsampled event tracks are used for inner product interval scoring (\cref{eqn:innerProductScoring}) to generate interval scoring matrices, which are then fed to the neural semi-CRF layer for log-likelihood calculation or inference.

\subsection{Rethinking Downsampling}

Existing studies on Vision Transformers (ViTs) demonstrate the effectiveness of {a non-hierarchical design that uses} highly downsampled, low-resolution feature maps even for tasks requiring dense predictions, e.g., \cite{Yanghao2022Plain}, {challenging the dominance of hierarchical models like UNET \cite{Ronneberger2015UNET}}.  
However, state-of-the-art (SOTA) piano transcription systems, including \cite{yan2021skipping,Toyama2023AutomaticPT, Kong2020HighResolutionPT, Hawthorne2021SequencetoSequencePT}, retain full resolution along the time axis. 
These approaches preserve the temporal detail of the input frames, but at the cost of increased training time and reduced model scalability.

This choice might be explained by concerns over losing temporal precision when locating events.
However, we argue that the high dimensionality of the embeddings makes the low temporal resolution feature map still capable of processing with enough information.

In our approach, we use strided convolutional layers to downsample the input spectrogram, along both the time and frequency axes, transforming it from its original spatial dimensions $(T, F)$ to a low-resolution feature map with dimensions $(T', F') = (\frac{T}{c_{T}}, \frac{F}{c_{F}})$. 
In line with the ViT literature, we refer to this reduced feature map as \textit{patch embeddings for $c_T\times c_F$ patches}. 
The choice of patch size $(c_T, c_F)$ may present a trade-off between computational efficiency and the model's capacity to capture dense events in the input spectrogram.
As an initial exploration, we use a patch size of $8\times 4$ to keep the training time within our expected range.

To upsample event tracks to the original temporal resolution of frames, we utilize a single transposed 1-d convolutional layer.
We found that this simple upsampling layer efficiently prepares representations for inner product scoring at the desired resolution.%

\subsection{Transformer Encoder Architecture}
\label{sec:encoder}
    
\begin{figure}[htb]
     \centering
     \vspace{-1em}
     \begin{subfigure}[t]{0.35\columnwidth}
         \centering
         \includegraphics[height=3.5cm,  trim={0 0.2cm 0.0cm 0.0cm},clip]{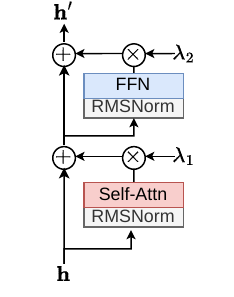}
         \caption{Transformer Block.}
         \label{fig:basicBlock}
     \end{subfigure}
     \hfill
     \begin{subfigure}[t]{0.35\columnwidth}
         \centering
         \includegraphics[height=3.5cm,  trim={0 0.2cm 0 0cm},clip]{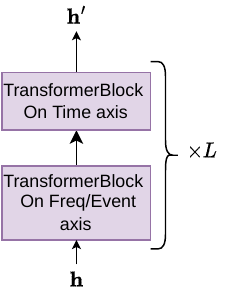}
         \caption{Encoder Layers}
         \label{fig:separableBlock}
     \end{subfigure}
        \caption{Building Blocks for the Transformer Encoder}
        \label{fig:buildingBlocks}
        \vspace{-1em}
\end{figure}

\paragraph{Spatial Position Embedding.}
We use learnable Fourier features for spatial position embeddings \cite{li2021learnable} for both time-frequency representations with coordinates $(\textit{frameIdx}, \textit{freqIdx})$, and event tracks with coordinates $(\textit{frameIdx}, \textit{eventTypeIdx})$. 
{This position embedding is chosen for its simplicity and broad compatibility with transformer architectures.}
Our formula differs slightly from \cite{li2021learnable} as we follow the formula in the original random Fourier features paper \cite{Rahimi2007RandomFF}.
We compute the position embedding $\vy \in \R^{E}$ from a multidimensional coordinate $\vx\in \R^{C}$ as: 
\begin{equation}
\begin{aligned}
\vy = g( \sqrt{\frac{2}{B}}\cos (\mW_r \vx + \vb) ),
\end{aligned}
\end{equation}
where $\mW_r$ is a learnable matrix $\R^{B \times C}$, initialized from $\normal (0, \gamma^{-2})$; $B$ is the dimension for the Fourier features; $\gamma$ is a hyperparameter; $\vb \in \R^{B}$ is the learnable bias term, initialized from $\mathcal{U}(-\pi, +\pi)$; $g: \R^{B} \to \R^{E}$ is a two-layer perceptron. 
This position embedding functions like an MLP that takes coordinates as input, with the first nonlinearity being a scaled cosine function.

\paragraph{The Transformer Encoder Layer.}
\Cref{fig:basicBlock} illustrates the basic transformer block.
This block first applies \textit{RMSNorm}\cite{zhang-sennrich-neurips19}  before the self-attention and feed-forward layers. 
To enhance training stability, we use ReZero \cite{pmlr-v161-bachlechner21a} which applies a learnable scaling factor $\lambda$, initially set to $0.01$, before adding to the skip connection. 
As in \Cref{fig:separableBlock}, for reducing computational cost, we alternate attention within each transformer block along the time and frequency/eventType axes; similar ideas are often used for efficient transformer architectures~\cite{Ho2019AxialAI, Ristea-INTERSPEECH-2022, Lu2024ROPE}.

\subsection{Segment-Wise Processing}
\label{subsec:segmentProcessing}
Longer audio is transcribed using segments with 50\% overlap.
Unlike~\cite{yan2021skipping}, which discards events that exceed the segment boundary during training, we truncate such events to fit within the segment.
We introduce two binary attributes, $\textit{hasOnset}$ and $\textit{hasOffset}$, to indicate whether an event's onset or offset has been truncated.

{
For each event type within a segment, decoding starts from either: (1) the current segment's boundary, or (2) the offset of the last event in the result set with $\textit{hasOffset}=\textit{true}$, whichever is later.
Events decoded in the current segment are then processed as follows: (1) non-overlapping events with $\textit{hasOnset}=\textit{true}$ are directly added to the result set; (2) for events overlapping with the last event of the same type in the result set: if the current event has $\textit{hasOnset}=\textit{true}$, it replaces the last event\footnotemark; otherwise, the two events are merged.
}

\footnotetext{{For overlapping events between segments: (1) The first event must have $\textit{hasOffset}=\textit{false}$. (2) A continuing second event must have $\textit{hasOnset}=\textit{false}$. (3) If the second event's $\textit{hasOnset}=\textit{true}$, the first event is replaced by the second event as it's not supported by the second.}}

\subsection{Attribute Prediction}
Attributes associated with each event include \textit{velocity}, \textit{refined onset/offset positions} (for dequantizing frame positions), and the binary flags $\textit{hasOnset}$ and $\textit{hasOffset}$. 
To predict these attributes for an event extracted from the event track $(\vh_i^{\text{eventType}})_{i=0}^{T-1}$, e.g., $[a, b]$, we use a two-layer MLP that takes $\vh_a^{\text{eventType}}$ and $\vh_b^{\text{eventType}}$ as input. 
The MLP outputs the parameters of the probability distributions for each attribute. Specifically, $\textit{velocity}\in \{0\ldots, 127\}$ is modeled as a categorical distribution, \textit{refined onset/offset positions} $\in (-0.5,0.5)$ are modeled as continuous Bernoulli distributions \cite{Loaiza2019CB} shifted by $-0.5$, and $\textit{hasOnset}/\textit{hasOffset} \in \{0,1\}$ are modeled as Bernoulli distributions.

\section{Experiment}

\subsection{Dataset} \label{sec:dataset}
\textbf{Maestro v3.0.0 \cite{hawthorne2018enabling}.}
This dataset contains about 200 hours of piano performances, including audio recordings and corresponding MIDI files captured using Yamaha Disklavier pianos. We use the standard train/validation/test splits.

\paragraph{MAPS \cite{Emiya2010MAPSA}.}
The MAPS dataset includes both synthesized and real piano recordings, with the real recordings captured by MIDI playback on Yamaha Disklavier. We evaluate our model on the Disklavier subset (ENSTDkAm/MUS and ENSTDkCl/MUS) of the MAPS dataset, which consists of 60 recordings and is commonly used for cross-dataset evaluation.
However, we discovered systematic alignment issues in the ground-truth annotations for both notes and pedals, affecting both onset and offset locations. 
Onset alignment issues have been previously reported in \cite{Cogliati2016CS} but are not widely known in the community\footnote{
A piece-dependent onset latency around 15 ms has been previously discussed in \cite{Cogliati2016CS}.
Due to the electro-mechanical playback mechanism, this latency could also be note/pedal dependent.
Offset deviation (up to approximately 70 ms) appears more complex and may be influenced by pedal-/note-dependent mechanical latency {or undocumented specific piano model's response to non-binary pedal values}. 
}.

\paragraph{SMD \cite{MuellerKBA11_SMD_ISMIR-lateBreaking}.} 
{Similar to Maestro dataset, the SMD dataset was created by recording human performance on a Yamaha Disklavier.}
We use SMD version 2. The dataset contains 50 recordings. We found that both the onset and offset annotations in SMD are better aligned compared to MAPS.

\subsection{Model Specification}

The key model specifications are summarized in \cref{table:modelSpec}.
Training takes about 6 days on 2 \textit{NVIDIA RTX 4090}. 

\setlength{\tabcolsep}{4pt}
\begin{table}[H]
\begin{adjustwidth}{-1.1in}{-1in} %
\begin{center}
\scriptsize
\begin{tabularx}{\columnwidth}{|l X|}
   \hline
   Input Mel Spectrogram& sr: 44100 Hz, hop: 1024, window size: 4096,  subwindows:5, mels: 229, freq: 30-8000 Hz, segment: 16s, \\
   \hline
   Patch & shape: $8\times 4$, embeding size: 256 \\
   \hline
   Strided Conv. Layers& initial proj. size: 64, added with freq. embeddings.    \\
   for Downsampling &out channels: [128, 256, 256, 256], kernel size: 3, strides: [(2,1), (2,2), (2,2), (1,1)],
    Each followed by GroupNorm, groups = 4, and GELU (except for the last conv.)\\
   \hline
   Position Embedding& $\gamma=1$, $|B| = 256$, MLP hidden size 1024 \\
   \hline
   Transformer Encoder& 8 heads, 6 layers (=12 blocks), FNN size: 1024\\
   \hline
   Upsampling  & 1d. transposed conv, out: 128, kernel size:8, stride:8 \\
   \hline
   Attribute Prediction & two layer MLP, hidden size: 512, dropout 0.1 \\
   \hhline{|=|=|}
   Batch Size & 12\\
   \hline
   Optimizer & Adabelief \cite{zhuang2020adabelief}, maximum learning rate: $4\mathrm{e}{-4}$ \\
   \hline
   Weight Decay & $1\mathrm{e}{-2}$, excluding bias, norm., and pos. embedding\\
   \hline
   Learning Rate Schedule & 500k iterations, $5\%$ warm-up phase, cosine anneal.\\
   \hline
   Gradient Clipping & Clipping norms at $80\%$ quantile of past 10,000 iterations\\
   \hline

\end{tabularx}
\end{center}
\end{adjustwidth}
\caption{Model Specification. }
\label{table:modelSpec}
\end{table}

\subsection{Evaluation Metrics}

We compute precision, recall, and f1 score averaged over recordings for both activation level (from \cite{yan2021skipping}, equivalent to frame level with infinitesimal hop size), and note level metrics (\textit{Note Onset}, \textit{Note w/Offset}, and \textit{Note w/Offset \& Vel.}, using \textit{mir\_eval}\cite{Raffel2014MIREval}, default settings).
All metrics are directly computed from transcribed MIDIs.
For details on these metrics, readers can refer to the supplementary material of \cite{yan2021skipping}, and the documentation of \textit{mir\_eval} \cite{craffel_mir_eval_transcription}.

Due to the ground-truth alignment issues discussed in \cref{sec:dataset} and space constraints, we only report activation-level and onset-only note-level metrics for MAPS and SMD.

\subsection{Results}

Our results on the Maestro v3 test set are presented in \cref{tab:noteResults}. 
The proposed model achieves state-of-the-art performance across all metrics in terms of f1 score, 
surpassing previous methods by a significant margin.
We also report results for soft pedal transcription which has not been previously explored.
The low event-level metrics suggest that accurately determining soft pedal onset and offset times is more challenging than for notes and sustain pedals. 
{We conjecture this is because soft pedals are typically engaged for longer durations and appear significantly less frequently in the dataset than sustain pedals.}

\paragraph{Scoring Methods Comparison.}  
We conducted an ablation study to compare our proposed inner product scoring with the more complex scoring method from \cite{yan2021skipping}.
We trained a model with an identical architecture but replaced the inner product scoring with the scoring module from \cite{yan2021skipping}. 
To ensure a fair comparison, we adjusted the hidden sizes of the scoring module to keep the training time for a single iteration within a factor of two of our proposed system. 
Specifically, all event tracks were projected to a single sequence with a dimension of 512, and the hidden size of the scoring module was set to 512. 
As shown in \cref{tab:noteResults}, our inner product scoring outperforms the more complex scoring method, demonstrating its effectiveness and efficiency.

Furthermore, we compared two variants of the inner product scoring: a linear layer and an MLP for computing the $\vk/\vq/b$ vectors ($f$ in \cref{eqn:kqbMapping}). The results demonstrate that the linear layer yields better performance than the MLP.
Interestingly, this aligns with how $\vk$ and $\vq$ are computed in transformers.

\paragraph{Effect of omitting incomplete events.} 
{
We found that omitting steps of handling incomplete events at segment boundaries (\cref{subsec:segmentProcessing}) only cause noticeable performance impact for pedals,  particularly the soft pedal (\cref{tab:noteResults}). 
This can be explained by the fact that pedal events, especially soft pedals, can often exceed the segment length, while notes are normally shorter than the segment length we choose.
}

\paragraph{Results on MAPS/SMD.} 
We evaluated our model on the MAPS dataset using three different ground-truth annotations: (1) Original, (2) Ad hoc Align, where the median deviation from the initial evaluation is subtracted from all notes for each piece and then re-evaluated, and (3) Cogliati, which subtracted a latency value per recording for ENSTDkCL as provided by \cite{Cogliati2016CS}. 
For the SMD dataset, only the original annotation is used. \Cref{tab:MAPSResults} presents the results.

All methods exhibit low activation-level F1 scores on MAPS. 
Using the onset-corrected annotation (Cogliati) on MAPS {increases the onset F1 score but} degrades the activation-level F1 score due to the uncorrected offset biases.
In fact, the Cogliati annotation achieves similar or lower activation-level F1 scores compared to all listed methods when evaluated against the original annotation.

All methods achieve F1 scores on SMD that are more comparable to those evaluated on Maestro. 
However, performance decreases significantly on MAPS, even with corrected annotations. 
This suggests that the dataset issue may be more complex than a simple piece-depedent timing shift.

Notably, the corrected annotations can lead to different conclusions compared to the original annotation. 
For example, while the data-augmented Onsets\&Frames model achieves a higher note onset F1 score than hFT using the original annotation, it scores lower than hFT when evaluated using the ad hoc  correction and the Cogliati annotation.

These observations highlight the need for caution when evaluating models on datasets created using mechanisms that may involve systematic biases, e.g., electromechanical playback. 
Despite these complications, our proposed system, with or without data augmentation\footnotemark, achieves the highest note onset F1 score among the compared methods on both SMD and MAPS with Ad hoc/Cogliati correction.
\footnotetext{Data augmentation: pitch shifting $\pm 20$ cents, adding noise from \cite{piczak2015dataset}, applying randomized 8 band EQ and impulse response from \cite{EchoThief}. }

\begin{savenotes}
\begin{table*}[htb]
\begin{adjustwidth}{-1in}{-1in} %
    \centering
    \scriptsize 
    \begin{tabular}{|r c ccc ccc ccc ccc|}
        \hline
        \textbf{Method} &  \textbf{\# Param%
        }
        &\multicolumn{3}{c}{\textbf{Activation}} & \multicolumn{3}{c}{\textbf{Note Onset}} & \multicolumn{3}{c}{\textbf{Note w/ Offset}} & \multicolumn{3}{c|}{\textbf{Note w/ Offset \& Vel.}} \\
        {} & &{P(\%)} & {R(\%)} & {$F_1$(\%)} & {P(\%)} & {R(\%)} & {$F_1$}(\%) & {P}(\%) & {R}(\%) & {$F_1$}(\%) & {P}(\%) & {R}(\%) & {$F_1$}(\%) \\
        \hline
{} & \multicolumn{13}{c|}{\textbf{Notes}} \\
\hline
SemiCRF \cite{yan2021skipping} 
&9.8M&
93.79&88.36&90.75&98.69&93.96&96.11&90.79&86.46&88.42&89.78&85.51&87.44\\
hFT, reported in \cite{Toyama2023AutomaticPT}
&5.5M&  92.82&93.66&93.24&\best{99.64}&95.44&97.44 & 92.52& 88.69& 90.53& 91.43&  87.67& 89.48
\\
hFT\cite{Toyama2023AutomaticPT} \setfootnotemark\label{tablenote:recompute} .
&5.5M&  95.37&90.82&92.93&\best{99.62}&95.41&97.43 & 92.22& 88.40& 90.23& 91.21&  87.44& 89.24
\\
\hline
Ours with scoring method in \cite{yan2021skipping}
&11.0M& 
93.79&92.40&93.06&98.61&95.92&97.23&91.69&89.23&90.43&91.08&88.64&89.83\\
Ours with MLP $\vk\vq b$  mapping 
&13.0M& 
95.66&94.79&95.20&99.54&96.91&98.19&94.39&91.92&93.12&93.84&91.40&92.59\\
Ours w/o incomplete events 
&12.9M& 
\best{93.76}&94.46&95.07&99.56&97.10&98.30&\best{94.66}&92.36&93.48&\best{94.12}&91.83&\best{92.95}\\
Ours
&12.9M&
{95.75}&\best{95.01}&\best{95.35}&99.53&\best{97.16}&\best{98.32}&94.61&\best{92.39}&\best{93.48}&{94.07}&\best{91.87}&{92.94}
\\
\hline
{} & \multicolumn{13}{c|}{\textbf{Sustain Pedals}} \\
\hline
Kong et al., reported in \cite{Kong2020HighResolutionPT}
&20.2M& 94.30&94.42&94.25&91.59&92.41&91.86&86.36&87.02&86.58
& - & - & - \\
Kong et al. \cite{Kong2020HighResolutionPT}$^{\ref{tablenote:recompute}}$  \setfootnotemark\label{tablenote:kong}
&20.2M& 94.14&94.29&94.11&77.43&78.19&77.71&73.56&74.21&73.81
& - & - & - \\
SemiCRF \cite{yan2021skipping}
&9.8M&
95.17&88.33&90.98&82.18&75.81&78.52&78.75&72.74&75.30 
& - & - & - \\
\hline
Ours w/o incomplete events 
& 12.9M& \best{96.69}&92.92&94.47& \best{89.10}& 83.96& 86.28& \best{86.33}&81.40&83.63 & - & - & - \\
Ours
& 12.9M& {96.67}&\best{94.46}&\best{95.40}&{88.96}&\best{84.22}&\best{86.37}&{86.19}&\best{81.66}&\best{83.71} & - & - & - \\
\hline

\hline
{} & \multicolumn{13}{c|}{\textbf{Soft Pedals}} \\
\hline
Ours w/o incomplete events& 12.9M 
&74.41&28.77& 36.54&20.24&9.08&11.69&17.19&7.51&9.76
& - & - & - \\
Ours& 12.9M 
&86.42&83.12&84.09&24.32&17.39&19.46&18.51&13.40&15.06
& - & - & - \\
\hline
 
\end{tabular}
\end{adjustwidth}
\caption{Transcription Result on Maestro v3.0.0 Dataset Test Split. }
\vspace*{-15pt}
\label{tab:noteResults}

\end{table*}
\end{savenotes}

\footnotetext[\getrefnumber{tablenote:recompute}]
              {
              Use their provided code and pretrained weights. Recomputed from transcribed MIDIs. 
              }
\footnotetext[\getrefnumber{tablenote:kong}]
             {
             Previous SOTA for sustain pedals. Their released code indicates a 200 ms onset tolerance for pedal evaluation, contrary to the reported 50 ms in their paper. Here, we use a 50 ms onset tolerance, which explains the large discrepancy between the numbers here and their reported results.
             }
             
\setlength{\tabcolsep}{2pt}
\begin{table}[htb]
\begin{adjustwidth}{-1.1in}{-1in} %
\begin{center}
    \scriptsize
    \begin{tabular}{|r c c ccc  ccc|}
       \hline
       \multicolumn{3}{|c}{} &
       \multicolumn{3}{c}{\textbf{Activation}} &
       \multicolumn{3}{c|}{\textbf{Note Onset}}   \\
       \hline
       \textbf{Method}&\textbf{Dataset}& \textbf{Groudtruth} 
       & {P}(\%) & {R}(\%) & {$F_1$}(\%)  
       & {P}(\%) & {R}(\%) & {$F_1$}(\%)  \\
        \hline
Onsets
&MAPS&Original & 90.27&80.33&84.87&87.40&85.56&86.41  \\
\&Frames\cite{hawthorne2018enabling} &MAPS&Ad hoc Align& 90.50&80.53&85.08&88.79&86.93&87.78 \\
w. Data Aug.$^{\ref{tablenote:recompute}}$&MAPS&Cogliati& 64.75&82.83&71.60&87.57&84.97&86.19 \\
\hline
hFT\cite{Toyama2023AutomaticPT}.$^{\ref{tablenote:recompute}}$
&MAPS&Original   & 91.53 & 71.03 & 79.81 & 84.63 & 85.75 & 85.13   \\
&MAPS&Ad hoc Align& 91.77 & 71.25 & 80.04& 87.32& 88.48 & 87.84 \\
&MAPS&Cogliati& 68.83 & 74.07& 70.24& 89.94&90.10& 89.97 \\
&SMD& Original &93.18&89.82&91.35&98.71&95.58&97.09 \\
\hline
Ours
&MAPS&Original&88.41&82.29&85.08&84.31&88.10&86.10 \\
&MAPS&Ad hoc Align &88.69&82.57&85.36&86.63&90.53&88.47 \\
&MAPS&Cogliati&65.74&84.69&72.78&89.60&91.39&90.44 \\

&SMD& Original & 92.36&95.24&93.73&98.16&97.65&97.89 \\
\hline
Ours
&MAPS&Original   &94.11&84.63&89.00&92.11&88.78&90.38 \\
w. Data Aug. &MAPS&Ad hoc Align& 94.35&84.84&89.22&94.21&90.76&92.41  \\
&MAPS&Cogliati&  67.77&87.39&75.03&94.66&91.43&92.98 \\
&SMD& Original & 93.38&95.91&94.57&99.77&97.68&98.70\\
\hline
\multicolumn{9}{|c|}{\textbf{Between Ground Truths}} \\
\hline
Cogliati \cite{Cogliati2016CS}
&MAPS& Original &98.86&69.22& 80.17 &100&100&100\\
\hline
\end{tabular}
\end{center}
\end{adjustwidth}
\caption{Transcription Result on MAPS and SMD. See Text for discussion of dataset issues.}
\label{tab:MAPSResults}
\end{table}

\section{Conclusion}
This paper introduces a simple and efficient method for scoring time intervals using scaled inner product operations for the neural semi-CRF framework for piano transcription. 
We demonstrate that the proposed scoring method is not only simple and efficient but also theoretically expressive for yielding the correct transcription result. 
Inspired by the similarity between the proposed scoring method and the attention mechanism, we employ a non-hierarchical, encoder-only transformer backbone to produce event track representations. 
Our method achieves state-of-the-art performance on the Maestro dataset across all subtasks.
Due to resource constraints, we have not evaluated the effect of patch and embedding sizes, which is left for future work. 
Additionally, future research could explore more advanced transformer architectures, investigate the interaction between transformer architecture and the neural semi-CRF layer, and extend the approach to other instruments and multi-instrument music transcription tasks.

\section{Acknowledgement}
{This work is supported in part by National Science Foundation (NSF) grants 1846184 and 2222129.}

\bibliography{main}

\appendix

\section{Changes from Previous Versions}\label{appendix:changes}

 We later discovered that the SMD version 1, which was used in the previous versions of the paper, missed all CC events, which affects all offsets of pedal-extended notes.
 We notified the original author and they provided the revised SMD version 2 \cite{SMDVersion2}. We updated the numbers in our paper to reflect this fix.

\end{document}

%% file: math_commands.tex
\usepackage{amsmath,amsfonts,bm}

\newcommand{\newterm}[1]{{\bf #1}}
\newcommand{\best}[1]{{\bf #1}}

\def\eqref#1{equation~\ref{#1}}

\def\1{\bm{1}}

\def\vone{{\bm{1}}}

\def\vb{{\bm{b}}}

\def\vh{{\bm{h}}}

\def\vk{{\bm{k}}}

\def\vm{{\bm{m}}}

\def\vq{{\bm{q}}}

\def\vx{{\bm{x}}}
\def\vy{{\bm{y}}}

\def\mK{{\bm{K}}}

\def\mQ{{\bm{Q}}}

\def\mS{{\bm{S}}}

\def\mW{{\bm{W}}}

\DeclareMathAlphabet{\mathsfit}{\encodingdefault}{\sfdefault}{m}{sl}
\SetMathAlphabet{\mathsfit}{bold}{\encodingdefault}{\sfdefault}{bx}{n}

\def\gX{{\mathcal{X}}}
\def\gY{{\mathcal{Y}}}

\def\sN{{\mathbb{N}}}

\newcommand{\R}{\mathbb{R}}

\DeclareMathOperator*{\argmax}{arg\,max}

\newcommand{\normal}{\mathcal{N}}